\RequirePackage{amsmath}
\documentclass[runningheads]{llncs}

\usepackage{amsmath}
\usepackage{amssymb}
\usepackage{graphics}
\usepackage{graphicx}
\usepackage[caption=false]{subfig}
\usepackage{tikz}
\usepackage{multirow}
\usetikzlibrary{positioning,chains,fit,shapes,calc,arrows,automata,shapes.geometric}

\newcommand{\mathsym}[1]{{}}

\newtheorem{fact}{Fact}

\newcommand{\ktc}{\texttt{k-terminal cut}}

\begin{document}

\title{Solving $(k-1)$-Stable Instances of \ktc{} with Isolating Cuts\thanks{The author is a Fellow of the National Physical Science Consortium.}}
\titlerunning{Stable Instances of \ktc{}}
%\thanks{Fellow, National Physical Science Consortium}

\author{Mark Velednitsky \inst{1} \orcidID{0000-0003-1176-5159}}
\authorrunning{M. Velednitsky}
\institute{University of California, Berkeley \\ \email{marvel@berkeley.edu}}

\date{\today}

\maketitle

%% Abstract
\begin{abstract}
The \ktc{} problem, also known as the Multiway Cut problem, is defined on an edge-weighted graph with $k$ distinct vertices called ``terminals.'' The goal is to remove a minimum weight collection of edges from the graph such that there is no path between any pair of terminals. The problem is NP-hard.

Isolating cuts are minimum cuts which separate one terminal from the rest. The union of all the isolating cuts, except the largest, is a $(2-2/k)$-approximation to the optimal \ktc{}. This is the only currently-known approximation algorithm for \ktc{} which does not require solving a linear program.

An instance of \ktc{} is $\gamma$-stable if edges in the cut can be multiplied by up to $\gamma$ without changing the unique optimal solution. In this paper, we show that, in any $(k-1)$-stable instance of \ktc{}, the source sets of the isolating cuts are the source sets of the unique optimal solution of that \ktc{} instance. We conclude that the $(2-2/k)$-approximation algorithm returns the optimal solution on $(k-1)$-stable instances. Ours is the first result showing that this $(2-2/k)$-approximation is an exact optimization algorithm on a special class of graphs.

We also show that our $(k-1)$-stability result is tight. We construct $(k-1-\epsilon)$-stable instances of the \ktc{} problem which only have trivial isolating cuts: that is, the source set of the isolating cuts for each terminal is just the terminal itself. Thus, the $(2-2/k)$-approximation does not return an optimal solution.
\end{abstract}

\section{Introduction}
The \ktc{} problem, also known as the Multiway Cut problem, is defined on an edge-weighted graph with $k$ distinct vertices called ``terminals.''  The goal is to remove a minimum weight collection of edges from the graph such that there is no path between any pair of terminals. The \ktc{} problem is known to be APX-hard \cite{dahlhaus1994complexity}.

In \cite{bilu2012stable}, Bilu and Linial introduced the concept of stability for graph cut problems. An instance is said to be $\gamma$-stable if the optimal cut remains uniquely optimal when every edge in the cut is multiplied by a factor up to $\gamma$. The concept of robustness in linear programming is closely related \cite{robinson1977characterization,ben2000robust}. Makarychev, Makarychev, and Vijayaraghavan \cite{makarychev2014bilu} showed that for $4$-stable instances of \ktc{}, the solution to a certain linear programming relaxation of the problem will necessarily be integer. Angelidakis, Makarychev, and Makarychev \cite{angelidakis2017algorithms} improved the result to $(2 - 2/k)$-stable instances using the same linear programming technique.

In an instance of \ktc{}, isolating cuts are minimum cuts which separate one terminal from the rest of the terminals. They can give useful information about the optimal solution: the source set of a terminal's isolating cut is a subset of that terminal's source set in an optimal solution \cite{dahlhaus1994complexity}. Furthermore, the union of all the isolating cuts, except for the cut with largest weight, is a $(2-2/k)$-approximation for the \ktc{} problem \cite{dahlhaus1994complexity}. This algorithm is the only currently-known approximation algorithm for \ktc{} which does not require solving a linear program \cite{cualinescu1998improved,karger2004rounding,buchbinder2013simplex,sharma2014multiway}. Thanks to their relative simplicity, isolating cuts are easily put into practice \cite{velednitsky2018isolation}. It is natural to wonder how the $(2-2/k)$-approximation performs on non worst-case instances.

In this paper, we establish a connection between isolating cuts and stability. We show that in $(k-1)$-stable instances of \ktc{}, the source sets of the isolating cuts equal the source sets of the unique optimal solution of that \ktc{} instance. It follows that the simple $(2-2/k)$-approximation of \cite{dahlhaus1994complexity} returns the optimal solution on $(k-1)$-stable instances. Ours is the first result showing that this $(2-2/k)$-approximation is an exact optimization algorithm on a special class of graphs.

Our result is tight. For $\epsilon > 0$, we construct $(k-1-\epsilon)$-stable instances of the \ktc{} problem which only have trivial isolating cuts: that is, the source set of the isolating cut for each terminal is just the terminal itself. In these $(k-1-\epsilon)$-stable instances, the $(2-2/k)$-approximation does not return an optimal solution.

In Section \ref{sect:preliminaries}, we introduce definitions and notation. In Section \ref{sect:proof}, we prove the main structural result, that in $(k-1)$-stable instance of \ktc{} the source sets of the isolating cuts equal the source sets of the optimal \ktc{}. In Section \ref{sect:limitations}, we construct a $(k-1-\epsilon)$-stable graph in which the source set of the isolating cut for each terminal is just the terminal itself.

\section{Preliminaries} \label{sect:preliminaries}
The notation $\{G = (V, E), w, T\}$ refers to an instance of the \ktc{} problem, where $G = (V, E)$ is an undirected graph with vertices $V$ and edges $E$. $T = \{t_1, \ldots, t_k\} \subseteq V$ is a set of $k$ terminals. The weight function $w$ is a function from $E$ to $\mathbb{R}^+$. 

For a subset of edges $E'\subseteq E$, the notation $w(E')$ is the total weight of edges in $E'$: 
$$w(E') = \sum_{e \in E'}{w(e)}.$$ 

For two disjoint subsets of vertices $V_1 \subseteq V$, $V_2 \subseteq V$, the notation $w(V_1, V_2)$ is the total weight of edges between $V_1$ and $V_2$: 
$$w(V_1, V_2) = \sum_{\substack{(v_1, v_2) \in E \\ v_1 \in V_1 \\ v_2 \in V_2}}{w((v_1, v_2))}.$$

We can further generalize this notation to allow for several disjoint subsets of vertices $V_1, \ldots, V_m \subseteq V$. In this case, we calculate the total weight of edges that go between two distinct subsets: 
$$w(V_1, \ldots, V_m) = \sum_i\sum_{j > i}{w(V_i, V_j)}.$$

For an instance $\{G = (V, E), w, T\}$ of the \ktc{} problem, we can refer to optimal solution in two equivalent ways. The first is in terms of the edges that are cut and the second is in terms of the source sets.

Referring to the optimal cut in terms of edges, we use the notation $E_\text{OPT}$: the subset of $E$ of minimum total weight whose removal ensures that there is no path between any pair of terminals. 

\emph{Source sets} are a partition of $V$ into $S_1, S_2, \ldots, S_k$ such that $t_i \in S_i$. We say that $S_i$ is the source set corresponding to $t_i$. We denote the optimal source sets $S_1^*, S_2^*, \ldots, S_k^*$.

The set of edges in the optimal cut is precisely the set of edges which go between distinct elements of the optimal partition $(S_1^*,\ldots,S_k^*)$. Combining the notation introduced in this section, $$w(E_\text{OPT}) = w(S_1^*, \ldots, S_k^*).$$

\subsection{Stability}
\begin{definition}[$\gamma$-Perturbation] \label{def:gammaperturbation}
Let $G = (V, E)$ be a weighted graph with edge weights $w$. Let $G' = (V, E)$ be a weighted graph with the same set of vertices $V$ and edges $E$ and a new set of edge weights $w'$ such that, for every $e \in E$ and some $\gamma > 1$,
$$w(e) \leq w'(e) \leq \gamma w(e).$$
Then $G'$ is a $\gamma$-perturbation of $G$.
\end{definition}

Stable instances are instances where the optimal solution remains uniquely optimal for any $\gamma$-perturbation of the weighted graph.

\begin{definition}[$\gamma$-Stability] \label{def:gammastability}
Let $\gamma > 1$. An instance $\{G = (V, E), w, T\}$ of \ktc{} is $\gamma$-stable if there is an optimal solution $E_\text{OPT}$ which is uniquely optimal for \ktc{} for every $\gamma$-perturbation of $G$.
\end{definition}

Note that the optimal solution need not be $\gamma$ times as good as \emph{any} other solution, since two solutions may share many edges. Given an alternative feasible solution, $E_\text{ALT}$, to the optimal cut, $E_\text{OPT}$, in a $\gamma$-stable instance, we can make a statement about the relative weights of the edges where the cuts differ. The following equivalence was first noted in \cite{makarychev2014bilu}:

\begin{lemma}[$\gamma$-Stability] \label{lemma:stability}
Consider an instance $\{G = (V, E), w, T\}$ of \ktc{} with optimal cut $E_\text{OPT}$. Let $\gamma > 1$. $G$ is $\gamma$-stable iff for every alternative feasible $k$-terminal cut $E_\text{ALT} \neq E_\text{OPT}$, we have $$w(E_\text{ALT} \setminus E_\text{OPT}) > \gamma w(E_\text{OPT} \setminus E_\text{ALT}).$$
\end{lemma}

\begin{proof}
Note that $E_\text{ALT}$ cannot be a strict subset of $E_\text{OPT}$ (since $E_\text{OPT}$ is optimal) and the claim is trivial if $E_\text{OPT}$ is a strict subset of $E_\text{ALT}$. Thus, we can assume that both $E_\text{ALT} \setminus E_\text{OPT}$ and $E_\text{OPT} \setminus E_\text{ALT}$ are non-empty.

For the ``if'' direction, consider an arbitrary $\gamma$-perturbation of $G$ in which the edge $e$ is multiplied by $\gamma_e$. We first derive the following two inequalities,

\begin{align*}
\sum_{e \in E_\text{OPT}}{\gamma_e w(e)} & = \sum_{e \in E_\text{OPT} \cap E_\text{ALT}}{\gamma_e w(e)} + \sum_{e \in E_\text{OPT} \setminus E_\text{ALT}}{\gamma_e w(e)} \\ & \leq \sum_{e \in E_\text{OPT} \cap E_\text{ALT}}{\gamma_e w(e)} + \gamma w(E_\text{OPT} \setminus E_\text{ALT}),
\end{align*}

and

\begin{align*}
\sum_{e \in E_\text{ALT}}{\gamma_e w(e)} & = \sum_{e \in E_\text{OPT} \cap E_\text{ALT}}{\gamma_e w(e)} + \sum_{e \in E_\text{ALT} \setminus E_\text{OPT}}{\gamma_e w(e)} \\ & \geq \sum_{e \in E_\text{OPT} \cap E_\text{ALT}}{\gamma_e w(e)} + w(E_\text{ALT} \setminus E_\text{OPT}).
\end{align*}

Since we have the inequality
$$w(E_\text{ALT} \setminus E_\text{OPT}) > \gamma w(E_\text{OPT} \setminus E_\text{ALT}),$$
we conclude that
$$\sum_{E_\text{OPT}}{\gamma_e w(e)} < \sum_{E_\text{ALT}}{\gamma_e w(e)}.$$
Hence, $E_\text{OPT}$ remains uniquely optimal in any $\gamma$-perturbation.

For the ``only if'' direction, if $G$ is $\gamma$-stable, then we can multiply each edge in $E_\text{OPT}$ by $\gamma$ and $E_\text{OPT}$ will still be uniquely optimal:

$$w(E_\text{ALT} \setminus E_\text{OPT}) + \gamma w(E_\text{ALT} \cap E_\text{OPT}) > \gamma w(E_\text{OPT} \setminus E_\text{ALT}) + \gamma w(E_\text{ALT} \cap E_\text{OPT}).$$

Thus,
$$w(E_\text{ALT} \setminus E_\text{OPT}) > \gamma w(E_\text{OPT} \setminus E_\text{ALT}).$$ \qed
\end{proof}

We make a few observations about $\gamma$-stability:

\begin{fact}
\label{fact:unique}
Any \ktc{} instance that is stable with $\gamma > 1$ must have a unique optimal solution.
\end{fact}

\begin{proof}
By Definition \ref{def:gammaperturbation}, any graph is a $\gamma$-perturbation of itself. Thus, by Definition \ref{def:gammastability}, the optimal solution must be unique. \qed
\end{proof}

\begin{fact}
Any \ktc{} instance that is $\gamma_2$-stable is also $\gamma_1$-stable for any $1 < \gamma_1 < \gamma_2$.
\end{fact}

\begin{proof}
The set of $\gamma_1$-perturbations is a subset of the set of $\gamma_2$-perturbations, since $$w(e) \leq w'(e) \leq \gamma_1 w(e) \implies w(e) \leq w'(e) \leq \gamma_2 w(e).$$ \qed
\end{proof}

Thus, for example, every instance which is $4$-stable is necessarily $2$-stable, but not the other way around.

\subsection{Isolating Cuts}
\begin{definition}[$t_i$-Isolating Cut]
The $t_i$-isolating cut is a minimum $(s,t)$-cut which separates source terminal $s = t_i$ from all the other terminals (shrunk into a single sink terminal $t = T \setminus \{t_i\}$). 
\end{definition}

We will use the notation $Q_i$ to denote the source set of this isolating cut (the set of vertices which remain connected to $t_i$). We use $E_i$ to denote the set of edges which are cut. Let $E_\text{ISO}$ be the union of the $E_i$ except for the $E_i$ with largest weight. The following lemmas are due to \cite{dahlhaus1994complexity}:

\begin{lemma} \label{lemma:approx}
$E_\text{ISO}$ is a $(2-2/k)$-approximation for the optimal $k$-terminal cut.
\end{lemma}

\begin{lemma} \label{lemma:isolation}
Let $\{G, w, T\}$ be an instance of \ktc{} and let $i \in \{1, \ldots, k\}$. Then there exists an optimal solution $(S_1^*,\ldots,S_k^*)$ in which $$Q_i \subseteq S_i^*.$$
\end{lemma}

The condition that ``there exists'' an optimal solution can make the implication of Lemma \ref{lemma:isolation} somewhat complicated when there are multiple optimal solutions, since the equation $Q_i \subseteq S_i^*$ need not be simultaneously true for all $i$. Conveniently, when an instance is $\gamma$-stable ($\gamma > 1$), it has a unique optimal solution (Fact \ref{fact:unique}). Thus, the condition $Q_i \subseteq S_i^*$ will be simultaneously true for all $i$.

\section{Proof of Main Result} \label{sect:proof}
\begin{theorem} \label{thrm:main}
Let $\{G, w, T\}$ be a $(k-1)$-stable instance of \ktc{}. Then, for all $i$, $Q_i = S_i^*$.
\end{theorem}

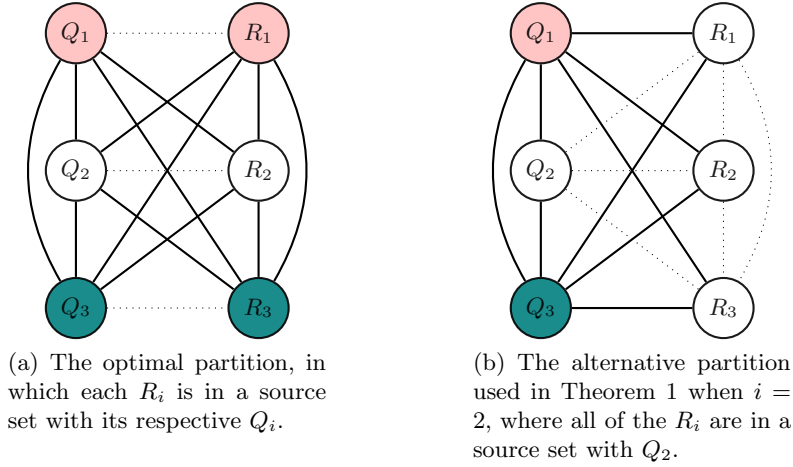
\begin{figure*}
    \centering
    \subfloat[][The optimal partition, in which each $R_i$ is in a source set with its respective $Q_i$. \label{fig:qr_opt}]{
        \begin{tikzpicture}[
        venn circle/.style = {circle, draw, thick, fill=#1,
                              minimum width=8mm, opacity=0.9},
        ]
            \node (t1) [venn circle = pink] {$Q_1$};
            \node (t2) [below = 1cm of t1, venn circle = white] {$Q_2$};
            \node (t3) [below = 1cm of t2, venn circle = teal] {$Q_3$};
            \node (s1) [right = 1.6cm of t1, venn circle = pink] {$R_1$};
            \node (s2) [below = 1cm of s1, venn circle = white] {$R_2$};
            \node (s3) [below = 1cm of s2, venn circle = teal] {$R_3$};
        
            \path[draw,dotted]
            (t1) edge node {} (s1)
            (t2) edge node {} (s2)
            (t3) edge node {} (s3)
            ;
            
            \path[draw,thick]
            (t1) edge node {} (s2)
            (t1) edge node {} (s3)
            (t2) edge node {} (s1)
            (t2) edge node {} (s3)
            (t3) edge node {} (s1)
            (t3) edge node {} (s2)
            ;
            
            \path[draw,thick]
            (s1) edge node {} (s2)
            (s2) edge node {} (s3)
            ;
            \path[draw,thick]
            (s1) edge [bend left=30] node {} (s3)
            ;
            
            \path[draw,thick]
            (t1) edge node {} (t2)
            (t2) edge node {} (t3)
            ;
            \path[draw,thick]
            (t1) edge [bend right=30] node {} (t3)
            ;

        \end{tikzpicture}
    }
    \qquad
    \qquad
    \qquad
    \subfloat[][The alternative partition used in Theorem \ref{thrm:main} when $i=2$, where all of the $R_i$ are in a source set with $Q_2$. \label{fig:qr_alt}]{
        \begin{tikzpicture}[
        venn circle/.style = {circle, draw, thick, fill=#1,
                              minimum width=8mm, opacity=0.9},
        ]
            \node (t1) [venn circle = pink] {$Q_1$};
            \node (t2) [below = 1cm of t1, venn circle = white] {$Q_2$};
            \node (t3) [below = 1cm of t2, venn circle = teal] {$Q_3$};
            \node (s1) [right = 1.6cm of t1, venn circle = white] {$R_1$};
            \node (s2) [below = 1cm of s1, venn circle = white] {$R_2$};
            \node (s3) [below = 1cm of s2, venn circle = white] {$R_3$};
        
            \path[draw,thick]
            (t1) edge node {} (s1)
            (t3) edge node {} (s3)
            ;
            
            \path[draw,thick]
            (t1) edge node {} (s2)
            (t1) edge node {} (s3)
            (t3) edge node {} (s1)
            (t3) edge node {} (s2)
            ;
        
            \path[draw,dotted]
            (t2) edge node {} (s1)
            (t2) edge node {} (s2)
            (t2) edge node {} (s3)
            ;

            \path[draw,dotted]
            (s1) edge node {} (s2)
            (s2) edge node {} (s3)
            ;
            \path[draw,dotted]
            (s1) edge [bend left=30] node {} (s3)
            ;
            
            \path[draw,thick]
            (t1) edge node {} (t2)
            (t2) edge node {} (t3)
            ;
            \path[draw,thick]
            (t1) edge [bend right=30] node {} (t3)
            ;
            
        \end{tikzpicture}
    }
    \caption{The sets $Q_1, Q_2, Q_3$ and $R_1, R_2, R_3$ defined in Theorem \ref{thrm:main} when $k=3$. Solid lines represent edges which are cut. Dashed lines represent edges which are not cut.}
    \label{fig:qr}
\end{figure*}

\begin{proof}
We will primarily be working with the $k$ vertex sets $Q_1, \ldots, Q_k$ and the $k$ vertex sets $S_1^*\setminus Q_1, \ldots S_k^*\setminus Q_k$. For convenience, we will use the notation $R_i = S_i^*\setminus Q_i$. As a consequence of Lemma \ref{lemma:isolation}, $S_i^* = Q_i \cup R_i$. We will assume, for the sake of contradiction, that at least one $R_i$ is non-empty.

Since $Q_i$ is the source set for the isolating cut for terminal $t_i$:
\begin{align*}
w(Q_i, V \setminus Q_i) &\leq w(S_i^*, V \setminus S_i^*) \\
w(Q_i, V \setminus Q_i) &\leq w(R_i, V \setminus S_i^*) + w(Q_i, V \setminus S_i^*) \\
- w(Q_i, V \setminus S_i^*) + w(Q_i, V \setminus Q_i)  &\leq w(R_i, V \setminus S_i^*) \\
w(Q_i, R_i) &\leq w(R_i, V \setminus S_i^*) \\
w(Q_i, R_i) &\leq \sum_{\{j | j \neq i\}}{w(R_i, R_j)} + \sum_{\{j | j \neq i\}}{w(R_i, Q_j)}
\end{align*}
Summing these inequalities over all the $i$:
\begin{align}
\sum_i{w(Q_i, R_i)} &\leq \sum_i{\sum_{\{j | j \neq i\}}{w(R_i, R_j)}} + \sum_i{\sum_{\{j | j \neq i\}}{w(R_i, Q_j)}} \nonumber \\
\sum_i{w(Q_i, R_i)} &\leq 2 w(R_1, \ldots, R_k) + \sum_i{\sum_{\{j | j \neq i\}}{w(R_i, Q_j)}} \label{eq:isocutcondition}
\end{align}

Next, we will consider alternatives to the optimal cut $(S_1^*,\ldots,S_k^*)$ and apply Lemma \ref{lemma:stability}. The optimal cut can be written as $$(S_1^*, \ldots, S_k^*) = (Q_1 \cup R_1, \ldots, Q_k \cup R_k).$$ We will consider alternative cuts $E_\text{ALT}^{(i)}$ where all the $R_j$ are in the same set of the partition, associated with $Q_i$. That is, we will consider 
$$\Big(S_1, \ldots, S_{i-1}, S_i, S_{i+1}, \ldots, S_k\Big) = \Big(Q_1, \ldots, Q_{i-1}, Q_i \cup (R_1 \cup \ldots \cup R_k), Q_{i+1}, \ldots, Q_k\Big).$$
See Figure \ref{fig:qr} for an illustration. We assumed that at least one of the $R_i$ is non-empty, so at least $k-1$ of these alternative cuts are distinct from the optimal one\footnote[1]{If only one $R_i$ is non-empty, then $E_\text{OPT} = E^{(i)}_\text{ALT}$ for this $i$. The corresponding inequality in Equation \ref{eq:applylemma} is not strict (both sides are $0$), but the other $k-1$ inequalities are strict and so the average (Equation \ref{eq:averagingresult}) is still a strict inequality.}. In order to apply Lemma \ref{lemma:stability}, we need to calculate $w(E_\text{OPT}\setminus E_\text{ALT}^{(i)})$ and $w(E_\text{ALT}^{(i)} \setminus E_\text{OPT})$.

To calculate $w(E_\text{ALT}^{(i)} \setminus E_\text{OPT})$, consider the edges in $E_\text{ALT}^{(i)}$ with one endpoint in $Q_j$ ($j \neq i$). The only edges which are \emph{not} counted in $E_\text{OPT}$ are those which go to $R_j$. Thus, $$w(E_\text{ALT}^{(i)} \setminus E_\text{OPT}) = \sum_{\{j | j \neq i\}}{w(R_j, Q_j)}.$$

To calculate $w(E_\text{OPT}\setminus E_\text{ALT}^{(i)})$, we must consider the set of edges which are in $E_\text{OPT}$ but not in $E_\text{ALT}^{(i)}$. For an edge not to be in $E_\text{ALT}^{(i)}$, it must be internal to one of the $Q_j$ ($j \neq i$) or internal to $Q_i \cup (R_1 \cup \ldots \cup R_k)$. None of the internal edges of the $Q_j$ are in $E_\text{OPT}$, so we need only consider the internal edges of $Q_i \cup (R_1 \cup \ldots \cup R_k)$: 
$$w(E_\text{OPT}\setminus E_\text{ALT}^{(i)}) = w(R_1, \ldots, R_k) + \sum_{\{j | j \neq i\}}{w(R_j, Q_i)}.$$

We apply Lemma \ref{lemma:stability}, with $\gamma = k-1$:
\begin{align}
(k-1) \cdot w(E_\text{OPT}\setminus E_\text{ALT}^{(i)}) &< w(E_\text{ALT}^{(i)} \setminus E_\text{OPT}) \label{eq:applylemma} \\
(k-1) \cdot w(R_1, \ldots, R_k) + (k-1) \cdot \sum_{\{j | j \neq i\}}{w(R_j, Q_i)} &< \sum_{\{j | j \neq i\}}{w(R_j, Q_j)}. \nonumber \\
\intertext{Averaging over the $k$ inequalities (one for each $i$)\footnotemark[1]:}
(k-1) \cdot w(R_1, \ldots, R_k) + \frac{k-1}{k} \sum_i{\sum_{\{j | j \neq i\}}{w(R_j, Q_i)}} &< \frac{k-1}{k} \sum_i{w(R_i, Q_i)}. \label{eq:averagingresult} \\
\intertext{We combine this with the inequality derived in Equation \ref{eq:isocutcondition}:}
(k-1) \cdot w(R_1, \ldots, R_k) + \frac{k-1}{k} \sum_i{\sum_{\{j | j \neq i\}}{w(R_j, Q_i)}} &< 2 \frac{k-1}{k} w(R_1, \ldots, R_k) + \frac{k-1}{k} \sum_i{\sum_{\{j | j \neq i\}}{w(R_i, Q_j)}}. \nonumber
\intertext{Notice that}
\sum_i{\sum_{\{j | j \neq i\}}{w(R_j, Q_i)}} &= \sum_i{\sum_{\{j | j \neq i\}}{w(R_i, Q_j)}}. \nonumber \\
\intertext{Therefore,}
(k-1) \cdot w(R_1, \ldots, R_k) &< 2 \frac{k-1}{k} w(R_1, \ldots, R_k). \nonumber
\end{align}

This is a contradiction, so it must be the case that $R_i = \emptyset$ for all $i$. Thus, $Q_i = S_i^*$ for all $i$. \qed
\end{proof}

\begin{corollary}
Let $\{G, w, T\}$ be a $(k-1)$-stable instance of \ktc{}. Then $E_\text{ISO}$ is the unique optimal solution to \ktc{} on this instance.
\end{corollary}

\section{Tightness of Main Result} \label{sect:limitations}
\begin{theorem} \label{thrm:limitations}
There exists a $(k-1-\epsilon)$-stable instance of \ktc{} for which $Q_i = \{t_i\} \neq S_i^*$ for all $i \in \{1, \ldots, k\}$.
\end{theorem}

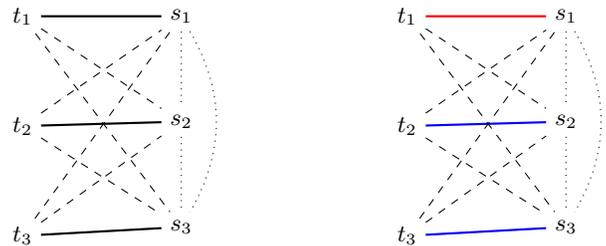
\begin{figure*}
    \centering
    \subfloat[][The dotted lines have weight $a$, solid have weight $b$, and dashed have weight $c$.\label{fig:k3example}]{
        \centering 
        \begin{tikzpicture}
            \node (t1) [] {$t_1$};
            \node (t2) [below = 1cm of t1] {$t_2$};
            \node (t3) [below = 1cm of t2] {$t_3$};
            \node (s1) [right = 1.6cm of t1] {$s_1$};
            \node (s2) [below = 1cm of s1] {$s_2$};
            \node (s3) [below = 1cm of s2] {$s_3$};
        
            \path[draw,thick]
            (t1) edge node {} (s1)
            (t2) edge node {} (s2)
            (t3) edge node {} (s3)
            ;
            
            \path[draw,dashed]
            (t1) edge node {} (s2)
            (t1) edge node {} (s3)
            (t2) edge node {} (s1)
            (t2) edge node {} (s3)
            (t3) edge node {} (s1)
            (t3) edge node {} (s2)
            ;
            
            \path[draw,dotted]
            (s1) edge node {} (s2)
            (s2) edge node {} (s3)
            ;
            
            \path[draw,dotted]
            (s1) edge [bend left=30] node {} (s3)
            ;
            
        \end{tikzpicture}        
    }
    \qquad
    \qquad
    \qquad
    \subfloat[][When $p=1$, we assume the red edges are in $E_\text{ALT}^{(1)}$ and the blue edge are not.\label{fig:p2example}]{
        \centering
        \begin{tikzpicture}
            \node (t1) {$t_1$};
            \node (t2) [below = 1cm of t1] {$t_2$};
            \node (t3) [below = 1cm of t2] {$t_3$};
            \node (s1) [right = 1.6cm of t1] {$s_1$};
            \node (s2) [below = 1cm of s1] {$s_2$};
            \node (s3) [below = 1cm of s2] {$s_3$};
        
            \path[draw,thick,blue]
            (t2) edge node {} (s2)
            (t3) edge node {} (s3)
            ;
            
            \path[draw,thick,red]
            (t1) edge node {} (s1)
            ;
            
            \path[draw,dashed]
            (t1) edge node {} (s2)
            (t1) edge node {} (s3)
            (t2) edge node {} (s1)
            (t2) edge node {} (s3)
            (t3) edge node {} (s1)
            (t3) edge node {} (s2)
            ;
            
            \path[draw,dotted]
            (s1) edge node {} (s2)
            (s2) edge node {} (s3)
            ;
            
            \path[draw,dotted]
            (s1) edge [bend left=30] node {} (s3)
            ;
        \end{tikzpicture}        
    }
    \caption{The construction used in Theorem \ref{thrm:limitations} when $k=3$.}
\end{figure*}

\begin{proof}
Consider a graph with $2k$ vertices. There are $k$ terminals ($t_1, \ldots, t_k$) and $k$ other vertices ($s_1, \ldots, s_k$). The $\binom{k}{2}$ edges between $s_i$ and $s_j$ ($i \neq j$) have weight $a \in \mathbb{R}^+$. The $k$ edges from $t_i$ to $s_i$ have weight $b \in \mathbb{R}^+$. The $k (k-1)$ edges from $t_i$ to $s_j$ ($i \neq j$) have weight $c \in \mathbb{R}^+$. Call this graph $G_k$. See Figure \ref{fig:k3example} for a drawing of $G_k$ when $k=3$. We will show that this graph has the desired properties for appropriate choices of $a, b, c \in \mathbb{R}^+$.

Consider an arbitrary $\gamma$. We would like to chose values of $a$, $b$, and $c$ such that the unique, $\gamma$-stable optimal cut is the one in which each $s_i$ remains connected to the corresponding $t_i$. Equivalently, zero edges with weight $b$ are cut. All the other edges \emph{are} cut. Thus, the optimal cut should have weight 
$$\binom{k}{2} a + k (k-1) c.$$

To verify that the optimal cut is $\gamma$-stable, we need to consider every other possible cut. We are helped by the symmetry of the construction. Consider an alternative cut, $E_\text{ALT}^{(p)}$, in which \textbf{exactly $p$ edges with weight $b$ are in the cut} (see figure \ref{fig:p2example}). Equivalently, exactly $k-p$ of the $s_i$ remain connected to the corresponding $t_i$. The optimal cut is the unique cut with $p=0$, so we need only consider alternative cuts where $p \in \{1, \ldots, k\}$. We would like to construct an inequality for the alternative cut of the form in Lemma \ref{lemma:stability}.

Consider $w(E_\text{ALT}^{(p)} \setminus E_\text{OPT})$. By construction, the number of edges of weight $b$ which are in $E_\text{ALT}^{(p)}$ but not $E_\text{OPT}$ is exactly $p$. Thus, 
$$w(E_\text{ALT}^{(p)} \setminus E_\text{OPT}) = p b.$$

Calculating $w(E_\text{OPT} \setminus E_\text{ALT}^{(p)})$ is more difficult. In order to create the tightest possible inequality in Lemma \ref{lemma:stability}, we want to include as few edges of weight $a$ and $c$ as possible in $E_\text{ALT}^{(p)}$ in order to maximize $w(E_\text{OPT}\setminus E_\text{ALT}^{(p)})$. We will consider the edges of weight $c$ and $a$ in the next two paragraphs.

Consider the edges of weight $c$. Notice that every $s_i$ is adjacent to all the terminals. Thus, if $s_i$ is one of the $k-p$ which remains connected to $t_i$, then all of the $k-1$ edges between $s_i$ and $t_j$ ($j \neq i$) must be in $E_\text{ALT}^{(p)}$. On the other hand, if the edge between $s_i$ and $t_i$ is in $E_\text{ALT}^{(p)}$, then at most one of the edges of weight $c$ adjacent to $s_i$ can be excluded from $E_\text{ALT}^{(p)}$. Thus, the number of edges of weight $c$ in $w(E_\text{OPT}\setminus E_\text{ALT}^{(p)})$ is at most $p$.

Consider the edges of weight $a$. Recall that $k-p$ of the edges with weight $b$ are \emph{not} in $E_\text{ALT}^{(p)}$, which means that there are $k-p$ vertices $s_i$ connected to the corresponding $t_i$. Between these $k-p$ vertices, all of the edges of weight $a$ must be in $E_\text{ALT}^{(p)}$. Thus, $\binom{k-p}{2}$ edges of weight $a$ must be in $E_\text{ALT}^{(p)}$. Of the $p$ vertices $s_j$ which are not connected to the corresponding $t_j$, each one can remain connected to at most one of the aforementioned $k-p$ vertices. When $p < k$, this gives an additional $p (k - p - 1)$ edges which must be in $E_\text{ALT}^{(p)}$.

Combining the arguments in the two preceding paragraphs, the strongest inequality we get from Lemma \ref{lemma:stability} when exactly $p$ edges of weight $b$ are cut is 
\begin{align}
p b &> \gamma \left(pc + (\binom{k}{2} - \binom{k-p}{2} - (k-p-1)p)a\right) &\text{if } p < k \nonumber \\
p b &> \gamma \left(pc + (\binom{k}{2})a\right) &\text{if } p = k. \nonumber \\
\intertext{Dividing both sides by $p$ and simplifying the coefficient of $a$:}
b &> \gamma \left(c + \frac{p+1}{2}a\right) &\text{if } p < k \nonumber \\
b &> \gamma \left(c + \frac{k-1}{2}a\right) &\text{if } p = k. \nonumber \\
\intertext{We need only consider the case $p = k-1$ to get the strongest possible inequality:}
b &> \gamma (c + \frac{k}{2} a). \label{eq:gamma_optimal_condition}
\end{align}

Computing the condition for the isolating cuts to have trivial source sets is easier. Knowing that $t_i$ and $s_i$ are connected in the optimal cut, we know that the optimal isolating cut can have source set either $\{t_i\}$ or $\{t_i, s_i\}$. The source set is $\{t_i\}$ if
\begin{align}
b + (k-1) c &< (k-1) a + 2 (k-1) c. \nonumber \\
b &< (k-1) (a + c).
\label{eq:isolation_condition}
\end{align}

In summary, if inequality \ref{eq:gamma_optimal_condition} is satisfied then $E_\text{OPT}$ is the $\gamma$-stable optimal \ktc{} and if inequality \ref{eq:isolation_condition} is satisfied then the isolating cuts have trivial source sets. When $\gamma = k - 1 - \epsilon$, the following values simultaneously satisfy inequalities \ref{eq:gamma_optimal_condition} and \ref{eq:isolation_condition}:
\begin{align*}
    a &= 2 \epsilon \\
    b &= k (k-1) (k-1-\epsilon) \\
    c & = k (k-1-\epsilon) - \epsilon.
\end{align*} \qed
\end{proof}

\section{Conclusions} \label{sect:conclusions}
In this paper, we proved that, in $(k-1)$-stable instances of \ktc{}, the source sets of the isolating cuts are the source sets of the unique optimal solution to that \ktc{} instance. As an immediate corollary, we concluded that the well-known $(2-2/k)$-approximation algorithm for \ktc{} is optimal for $(k-1)$-stable instances.

We also showed that the factor of $k-1$ is tight. We constructed $(k-1-\epsilon)$-stable instances of \ktc{} in which the source set of the isolating cut for a terminal is just the terminal itself. In those instances, the $(2-2/k)$-approximation algorithm does not return an optimal solution.

\bibliographystyle{splncs04}
\bibliography{ref}

\end{document}